\documentclass[a4paper,UKenglish,cleveref, autoref, thm-restate]{lipics-v2021}

\DeclareMathOperator*{\R}{\mathbb{R}}

\newtheorem{fact}[theorem]{Fact}

\title{A Local Search Algorithm for the Min-Sum Submodular Cover Problem} 

\titlerunning{Local Search for Min-Sum Submodular Cover} 

\author{Lisa Hellerstein}{Department of Computer Science and Engineering, NYU Tandon, New York, USA \and \url{https://cse.engineering.nyu.edu/~hstein/} }{lisa.hellerstein@nyu.edu}{https://orcid.org/
0000-0002-3743-7965}{}

\author{Thomas Lidbetter}{Department of Engineering Systems and Environment, University of Virginia, Charlottesville, Virginia, USA \\ Department of Management Science and Information Systems, Rutgers Business School, Newark, New Jersey, USA \and \url{https://engineering.virginia.edu/faculty/thomas-lidbetter}}{tlidbetter@virginia.edu}{https://orcid.org/0000-0001-6111-2899}{}

\author{R. Teal Witter}{Department of Computer Science and Engineering, NYU Tandon, New York, USA \and \url{https://www.rtealwitter.com} }{rtealwitter@nyu.edu}{https://orcid.org/0000-0003-3096-3767}{}

\authorrunning{L. Hellerstein, T. Lidbetter, and R.\,T. Witter} 

\Copyright{Lisa Hellerstein, Thomas Lidbetter, and R. Teal Witter} 



\category{} 

\relatedversion{} 

\supplement{}

\funding{Lisa Hellerstein was supported in part by NSF Award IIS-1909335.
Thomas Lidbetter was supported in part by NSF Award IIS-1909446.}

\acknowledgements{We would like to thank Christopher Musco for pointing out that the set function induced by entropy on continuous domains is not submodular.}

\nolinenumbers 

\EventEditors{Sang Won Bae and Heejin Park}
\EventNoEds{2}
\EventLongTitle{33rd International Symposium on Algorithms and Computation (ISAAC 2022)}
\EventShortTitle{ISAAC 2022}
\EventAcronym{ISAAC}
\EventYear{2022}
\EventDate{December 19--21, 2022}
\EventLocation{Seoul, Korea}
\EventLogo{}
\SeriesVolume{248}
\ArticleNo{13}

\begin{document}

\maketitle

\begin{abstract}
We consider the problem of solving the Min-Sum Submodular Cover problem using local search. The Min-Sum Submodular Cover problem generalizes the NP-complete Min-Sum Set Cover problem, replacing the input set cover instance with a monotone submodular set function. A simple greedy algorithm achieves an approximation factor of 4, which is tight unless P=NP [Streeter and Golovin, NeurIPS, 2008]. We complement the greedy algorithm with analysis of a local search algorithm. Building on work of Munagala et al. [ICDT, 2005], we show that, using simple initialization, a straightforward local search algorithm achieves a $(4+\epsilon)$-approximate solution in time $O(n^3\log(n/\epsilon))$, provided that the monotone submodular set function is also second-order supermodular. Second-order supermodularity has been shown to hold for a number of submodular functions of practical interest, including functions associated with set cover, matching, and facility location. We present experiments on two special cases of Min-Sum Submodular Cover and find that the local search algorithm can outperform the greedy algorithm on small data sets.
\end{abstract}

\section{Introduction}


We consider the Min-Sum Submodular Cover problem, defined as follows.
The input to the problem consists of an oracle for a monotone submodular function $u:2^{[n]} \rightarrow \R_{\geq 0}$, and positive costs $c_1, \ldots, c_n \in \R_{>0}$, where $[n] = \{1, \ldots, n\}$.
Let $c:2^{[n]} \rightarrow \R_{\geq 0}$ be such that for all $S \subseteq [n]$, $c(S) = \sum_{i \in S} c_i$.
We refer to $u$ as the ``utility'' and $c$ as the ``cost'' function.
The problem is to find
the permutation of the elements of $[n]$ that minimizes 
\begin{align}\label{eq:mspp}
  \sum_{i=1}^{n} c(S_i) \left(u(S_i)-u(S_{i-1})\right) 
\end{align}
where $S_i$ is the set containing the first $i$ elements of the permutation.
The Min-Sum Submodular Cover problem
generalizes the NP-Complete Min-Sum Set Cover problem introduced by Feige et al.~\cite{feige2004approximating}.
It has a simple greedy algorithm that achieves a 4-approximation~\cite{iwataetal12,streetergolovin}.\footnote{See the Appendix \ref{app:prior} for comments on Streeter and Golovin ~\cite{streetergolovin}.}  The 4-approximation is tight, assuming $P \neq NP$~\cite{feige2004approximating}.

In this work, we analyze a local search
algorithm for Min-Sum Submodular Cover.
Local search algorithms have been extensively applied
to discrete optimization problems \cite{antunes_et_al:LIPIcs:2017:7829, kutiel_et_al:LIPIcs:2017:7821, li2014deeper}
and offer several benefits over other types
of algorithms \cite{aarts2003local}.
One advantage of local search algorithms compared
to greedy methods in practice is their ability to explore a diverse
set of solutions.
In Section~\ref{sec:experiments}, we present the results of preliminary experiments which demonstrate that this ability can yield improved solutions to Min-Sum Submodular Cover instances.  

The local search algorithm we consider works 
in iterations, starting with an initial solution. 
In an iteration, the algorithm updates the current
solution to its best ``neighbor''.
Once a solution is locally optimal (or after
a fixed number of iterations),
the algorithm returns the current solution.

Our analysis builds on previous work of 
Munagala et al.~\cite{munagala2005pipelined}
for the Pipelined
Set Cover problem (defined below).
Munagala et al.\
exploit the observation that utility can be attributed
to elements of the ground set for the set cover instance, and use these elements
as variables in their linear program.
The main challenge of generalizing their analysis
is that the utility in general submodular set
functions is more abstract and cannot be attributed
to particular objects.
In order to apply the linear program used
in Munagala et al.,
our analysis  relies on an additional
property of the utility function called second-order 
supermodularity.
We leave as an open question whether the algorithm
gives a $(4+\epsilon)$-approximation
even without second-order supermodularity.

Second-order supermodularity was first studied
by Korula et al.~\cite{korula2018online}.
It can be viewed as
a natural extension of submodularity:
If one considers the multilinear extension
$F:[0,1]^n \rightarrow \R_{\geq 0}$ of
a set function $f:\{0,1\}^n \rightarrow \R_{\geq 0}$
(which is a way of interpolating the values of $f$
from the vertices of the Boolean hypercube to points in its interior),
the submodularity of $f$ is equivalent
to the property that the second partial
derivatives of $F$ are non-positive.
As mentioned by Korula et al.~\cite{korula2018online}
and Iyer et al.~\cite{iyer2021generalized},
the second-order supermodularity of $f$
is equivalent to the property that the
third partial derivatives of $F$ are non-negative.

The second-order supermodularity property is not overly restrictive;
there are several classes of submodular functions that
have this property including
weighted coverage functions,
weighted matching functions, and facility location \cite{korula2018online, iyer2021submodular, iyer2021generalized}.
Since this property was first defined, improved bounds  have been obtained for optimization problems by assuming the property
\cite{korula2018online, iyer2021submodular}.
The related properties of second-order modularity and second-order submodularity
have also been used in analyzing 
local search algorithms for constrained submodular 
maximization problems
\cite{ghadiri2019beyond, ghadiri2020parameterized}.

We now define two special cases of Min-Sum Submodular Cover: the
Pipelined Set Cover problem and the Min-Sum Facility Location problem.

\paragraph*{Pipelined Set Cover}
The inputs to the problem consist of
(i) $m$ ``ground'' elements $\{1,\ldots,m\}=[m]$,
(ii) $D_1, \ldots, D_n$, a family of $n$ subsets of the ground
elements $[m]$ such that $\bigcup_{i \in [n]} D_i = [m]$, and
(iii) positive costs $c_1, \ldots, c_n$ associated 
with each $D_i$.
Let $u:2^{[n]} \rightarrow \R_{\geq 0}$ be such that for all
$S \subseteq [n]$, $u(S)$ is the number 
of ground elements in $\bigcup_{i \in S} D_i$.
We call $u$ a ``coverage'' function.
Let $c:2^{[n]} \rightarrow \R_{\geq 0}$ be such that 
for all $S \subseteq [n]$, $c(S) = \sum_{i \in S} c_i$.
The problem is to find the permutation of
$[n]$ that minimizes the objective function in the Min-Sum Submodular Cover problem, 
$\sum_{i=1}^{n} c(S_i) \left(u(S_i)-u(S_{i-1})\right)$.
Thus Pipelined Set Cover problem is equivalent to the special case of the Min-Sum Submodular Cover problem where the utility function $u$ is a coverage function.  
The Min-Sum Set Cover problem is the special case of Pipelined Set Cover with unit costs.  (We note that Munagala et al. also present results for Weighted Pipelined Set Cover, where the ground elements have weights.)

\paragraph*{Min-Sum Facility Location}

Consider the following problem facility location problem, 
studied by Krause and Golovin~\cite{krausegolovinsurvey}.
There is a set $[n]$ of possible locations where facilities 
could be opened, to serve a collection of $m$ customers.
Opening a facility at location $a$ provides a service of 
value $M_{a,b}$ to customer $b$, where $M \in \R_{\geq 0}^{n \times m}$.
The utility of opening facilities in a subset $S$ of the locations is $u(S)$, where $u(S) =  \sum_{b=1}^m\max_{a \in S} M_{a,b}$.
This corresponds to the total value obtained by all the customers, assuming each customer chooses the open facility with highest service value.
The problem of Krause and Golovin is to maximize the utility function $u$ subject to a constraint on the number of facilities that can be opened.

We introduce a min-sum version of this facility location problem by considering the Min-Sum Submodular Cover problem with the utility function $u$ just described, and with $c_i$ representing
the time to open a facility $i$. 
This problem corresponds to a situation where facilities 
will be opened in all $n$ locations,
but they can only be opened one at a time.
$M_{a,b}$ represents the estimated value facility $a$ will provide to customer $b$ per unit of time, once facility $a$ is opened.  
Minimizing the objective value $\sum_{i=1}^{n} c(S_i) \left(u(S_i)-u(S_{i-1})\right)$ corresponds to finding the 
order to build facilities so as to minimize
lost value as facilities are built.

\paragraph*{Our Contributions}
We introduce the study of 
solving Min-Sum Submodular Cover using local search.
Building on work of Munagala et al.~\cite{munagala2005pipelined}, who presented a local-search algorithm for Pipelined Set Cover, we generalize their LP-based analysis by redefining a key quantity in their proof and using second-order supermodularity.
We show that local search produces a $(4+\epsilon)$-approximate solution for Min-Sum Submodular Cover
in time $O(n^3 \log (\frac{d}{\epsilon}))$, assuming second-order supermodularity of the utility function,  when initialized with a $d$-approximate solution.
We prove that a permutation listing the items in non-decreasing cost order is an $n$-approximate solution.
Thus initializing local search with a non-decreasing cost permutation
enables us to reach a $(4+\epsilon)$-approximate solution in time
$O(n^3 \log (\frac{n}{\epsilon}))$.  Applying this result
to Pipelined Set Cover improves on
the $O(n^3 \log (\frac{mn}{\epsilon}))$ time bound from Munagala et al., where $m$ is the size of the ground set of the set cover instance, by eliminating the dependence on $m$.
We also present results of experiments on two types of
Min-Sum Submodular Cover problems: Pipelined Set Cover
and Min-Sum Facility Location.
Our empirical findings suggest that local search
can reliably produce better solutions than
the natural greedy algorithm on small data sets.

\section{Preliminaries}\label{sec:prelim}
Let $f(e|S) := f(S\cup\{e\})-f(S)$
be the marginal utility of 
adding element $e$ to set $S$.
With this notation in hand,
we define several useful properties of set functions.

\begin{definition}[Set Function Properties]
\label{def:properties}
Consider a positive integer $n$ and
set function $f:2^{[n]} \rightarrow \R_{\geq 0}$.
We first define the following properties of $f$, which
hold if the inequality given below for that property holds for
all $S \subseteq [n]$ and all $i, j, k \in [n] \setminus S$,

\begin{itemize}
    \item \textit{monotone}:  $f(S \cup \{i\}) \geq f(S),$
    \item \textit{submodular} (diminishing returns):
    $f(i|S) \geq f(i|S \cup \{j\}),$
    
    \item \textit{second-order supermodular}:
    $f(i|S) - f(i|S \cup \{j\})
    \geq f(i|S \cup \{k\}) - f(i|S \cup \{k,j\}).$
    
\end{itemize}


\end{definition}

Note that the way in which we have written the above properties illustrates the progression
from monotonicity to submodularity and
submodularity to second-order
supermodularity:
we arrive at the `next' property by subtracting
the left-hand side from the right-hand side.
Another related property is modularity: for all $S \subseteq [n]$, $f(S) = \sum_{i \in S} f(\{i\})$.
In this paper, the functions we consider
will be monotone set functions that are normalized, i.e.,
$f(\emptyset)=0$ unless otherwise stated.



The Min-Sum Submodular Cover problem
is a special case of the Min-Sum
Permutation Problem,
defined by Happach et 
al.~\cite{happach2020general}.
That problem has the same
objective function as Min-Sum
Submodular Cover, and minimization may be over all permutations, or only over a subset of them.
The only assumptions on $u$ and $c$ in \cite{happach2020general}
are that they are monotone and
normalized.

\section{A Local Search Algorithm for Min-Sum Submodular Cover}\label{sec:main}

Munagala et al.~\cite{munagala2005pipelined} gave a local search algorithm for the special case of the Min-Sum Submodular Cover problem where $u$ is a coverage function.
Applying the same approach to the general Min-Sum Submodular Cover problem, we have the following local search algorithm:
initialize the algorithm with a permutation $\pi$ of $[n]$. 
Define a neighbor $\pi'$ of $\pi$ to be a permutation that
can be be produced from $\pi$ by
removing the element in some position $i$ of $\pi$ and reinserting it in position $j$.
Find the
neighbor $\pi'$ of $\pi$ with lowest 
objective value
(given by Equation \ref{eq:mspp}).
If that value is less than 
the objective value of $\pi$,
then replace $\pi$ by $\pi'$ and repeat.
Otherwise, output $\pi$.
Pseudocode for this algorithm is given in Algorithm~\ref{alg:local}.

The analysis of
Munagala et al.~\cite{munagala2005pipelined} shows that in the special case where $u$ is a coverage function,
Algorithm~\ref{alg:local} 
achieves a $(4+\epsilon)$-approximation to the optimal
permutation.
We
generalize their analysis
to all utility functions $u$ that are submodular and
second-order supermodular (in addition to being monotone and normalized, which we assume is the case for all utility functions in this paper).
%
%

\begin{lstlisting}[caption={Local search algorithm to produce a $(4+\epsilon)$-approximation.},label=alg:local,captionpos=t,float,abovecaptionskip=-\medskipamount, mathescape=true,escapeinside={*}{*}]
    Input: $\epsilon>0$, $n > 0$, utility function $u:2^{[n]}\rightarrow \R_{\geq 0}$,
    cost function $c:2^{[n]}\rightarrow \R_{\geq 0}$, $d$-approximate permutation $\pi$
    Output: permutation $\pi$
    for iteration in $\{1,\ldots,2n^3 \log(d/\epsilon)\}$ do
        $\pi^* \gets \pi$
        for $i,j \in [n]$ do
            # $\pi'$ is $\pi$ with $\pi(i)$ moved to position $j$
            $\pi' \gets move(\pi, i, j)$ 
            # objective$(u,c,\pi)$ is *Equation~(\ref{eq:mspp})*
            if objective$(u,c,\pi') < $ objective$(u,c,\pi^*)$ do
                $\pi^* \gets \pi'$ 
        if $\pi^* = \pi$
            # Algorithm converged:
            # $\pi$ is locally optimal with  respect to moves
            return $\pi$ # $4$-approximation
        $\pi \gets \pi^*$
    return $\pi$ # $(4+\epsilon)$-approximation
\end{lstlisting}

Let $\pi_{c}$ be a non-decreasing
cost permutation, i.e., $c(\{\pi_c(i)\}) \leq c(\{\pi_c(i+1)\})$
for $i \in [n-1]$.
We prove the following results.

\begin{theorem}\label{thm:4approx}
Fix a positive integer $n$.
Let $u:2^{[n]} \rightarrow \R_{\geq 0}$ be a submodular
and second-order
supermodular set function
and let $c:2^{[n]} \rightarrow \R_{\geq 0}$ be a modular set function.
If Algorithm~\ref{alg:local} converges before terminating, then the solution it returns
is a $4$-approximation to Min-Sum Submodular Cover on $u$
and $c$.
\end{theorem}

Unfortunately, we cannot guarantee that Algorithm~\ref{alg:local}
will converge before terminating.
The next result guarantees a $(4+\epsilon)$-approximation
when the algorithm terminates.

\begin{theorem}\label{thm:4epsapprox}
Consider the positive integer $n$, utility function $u$,
and cost function $c$ considered in Theorem~\ref{thm:4approx}.
Fix $\epsilon > 0$.
Let $\pi$ be a $d$-approximate permutation. 
If Algorithm~\ref{alg:local} does not converge before terminating, then the solution it returns (after $2n^3 \log(d/\epsilon)$ iterations), 
is a $(4+\epsilon)$-approximation to Min-Sum Submodular
Cover on $u$ and $c$.
\end{theorem}

Assuming constant access query access to $u$ and $c$,
Algorithm~\ref{alg:local} returns a $(4+\epsilon)$-approximation
in $2n^3 \log(d/\epsilon)$ time.


As in Munagala et al.~\cite{munagala2005pipelined},
in our analysis
we consider
a modified version of local search
based on ``insertions'' rather than ``moves.'' 
We find it easier to analyze local search
with insertions and the approximation result
immediately applies to local search with moves
since a permutation that is locally optimal
with respect to moves is also locally optimal
with respect to insertions.
In each iteration of local search with insertions, rather than considering the set of neighbors $\pi'$ of $\pi$,
the modified algorithm considers a set of what we will call pseudo-neighbors.  
Each is derived from $\pi$ by taking an element appearing in some position $i$ of $\pi$, and inserting a {\em second copy} of the element into some position $j < i$.  
Each pseudo-neighbor of $\pi$ corresponds to a unique neighbor of $\pi$, produced from the pseudo-neighbor by removing the original copy of the repeated element (which appears closer to the end of the permutation).

Define the objective value of a pseudo-neighbor $\pi'$ (which has length $n+1$)
to be $\sum_{i=1}^{n+1} c'(S_i)[u(S_i)-u(S_{i-1})]$, where here $S_i$ is the prefix of $\pi'$ containing its first $i$ elements, $u(S_i)$ is the value of $u$ for the set of {\em distinct} items in $S_i$, and $c'(S_i) = \sum_{j=1}^i c(\{s_j\})$ where $s_j$ is the element in position $j$ of $\pi'$.  That is, if both copies of the repeated element appear within the first $i$ positions of $\pi'$, then $c'(S_i)$ charges for both copies.

If the objective value of $\pi$ is no greater than the value of its pseudo-neighbors, then the modified algorithm outputs $\pi$.
Otherwise, the algorithm takes the pseudo-neighbor with lowest objective value,
deletes the original copy of its repeated element, and uses the resulting permutation as
the new value of 
$\pi$ in the next iteration.

The objective value of a pseudo-neighbor of $\pi$ is clearly greater than or equal to the objective value of the corresponding neighbor.  
Therefore, if $\pi$ has no neighbor with lower objective value, then it has no pseudo-neighbor with lower objective value.
It follows that the bounds we prove on the modified local search algorithm (with insertions) also apply to the original local search algorithm (with moves).

We prove Theorems \ref{thm:4approx} and \ref{thm:4epsapprox}
in the remainder of this section.

\subsection{Proof of Theorem \ref{thm:4approx}: $4$-approximation}\label{sub:4approx}

Say a permutation $\pi$ is locally optimal if no pseudo-neighbor has lower objective value.
We begin by proving that a locally optimal solution satisfies a certain inequality,  expressed in terms of variables $b_{ij}$.
This inequality is taken from the analysis in Munagala et al.~\cite{munagala2005pipelined}, but we define the variables $b_{ij}$ differently here.
We will use the following technical
Observation~\ref{obs:ineq} to prove
Lemma~\ref{lemma:localrule}.
\begin{observation}\label{obs:ineq}
   Consider three sequences of non-negative real numbers, $X_0, \ldots, X_n$, $Y_0, \ldots, Y_n$, and $C_0, \ldots, C_n$.  Let $j \in [n]$.
   Suppose that the following hold: (i) $X_0 \geq Y_0$, (ii) for all $r \in \{j, \ldots, n\}$,  
    $C_0 \leq C_r$ and
    $X_r \leq Y_r$, 
    and 
    (iii) $X_0 + \sum_{r=j}^n X_r
    = Y_0 + \sum_{r=j}^n Y_r$.
   Then 
   \begin{align*}
       C_0 X_0 + \sum_{r=j}^n C_r X_r
       \leq C_0 Y_0 + \sum_{r=j}^n C_r Y_r.
   \end{align*}
\end{observation}

\begin{proof}
    Rewriting the final assumption yields
    \begin{align*}
        X_0 - Y_0 = \sum_{r=j}^n (Y_r - X_r)
        \iff C_0(X_0 - Y_0) 
        = \sum_{r=j}^n C_0 (Y_r - X_r) \\
        \Rightarrow C_0(X_0 - Y_0) \leq
        \sum_{r=j}^n C_r (Y_r - X_r)
    \end{align*}
    where the second implication follows
    from the non-negativity of $Y_r - X_r$
    and the assumption that $C_0 \leq C_r$.
    Observation \ref{obs:ineq} follows
    immediately.
\end{proof}

\begin{lemma}\label{lemma:localrule}
    Suppose $u:2^{[n]} \rightarrow \R_{\geq 0}$
    is a submodular and second-order supermodular
    set function and $c:2^{[n]} \rightarrow \R_{\geq 0}$
    is a modular set function.
    Let $L_j$ denote the first $j$ elements
    of the locally optimal permutation
    and $O_i$ denote the first $i$ elements
    of the optimal permutation.
    Similarly, we use $l_j$ to represent the $j$th element
    of the local permutation
    and $o_i$ to represent the $i$th element
    of the optimal permutation.
    Then
    \begin{align}\label{eq:constraint}
    \sum_{r=j}^n c(L_r) \sum_{s=1}^n b_{sr}
    \leq 
    [c(o_i) + c(L_{j-1})] \sum_{r=j}^n b_{ir}
    +\sum_{r=j}^n [c(o_i) + c(L_r)],
    \sum_{\substack{s=1\\s\neq i}}^n b_{sr}
    \end{align}
    where 
    \begin{align*}
     b_{ij} = u(o_i|O_{i-1} \cup L_{j-1})
    - u(o_i|O_{i-1} \cup L_{j-1} \cup \{l_j\}) \\  
    = u(l_j|O_{i-1} \cup L_{j-1})
    - u(l_j|O_{i-1} \cup L_{j-1} \cup \{o_i\}).
    \end{align*}
\end{lemma}

\begin{proof}
When $u$ is a coverage function,
as in the analysis in Munagala et al.~\cite{munagala2005pipelined},
$b_{ij}$ represents the number of ground
elements covered in the optimal permutation by subset $o_i$
(and not by $o_1, \ldots, o_{i-1}$)
and in the local permutation
by subset $l_j$ (and not by $l_1, \ldots, l_{j-1})$.
Our definitions of $b_{ij}$
generalize this intuition to functions
where the utility is more abstract.
In particular, we can use telescoping sums
to derive the following identities:
\begin{align}
    &\sum_{i=1}^n b_{ij} = u(l_j|L_{j-1})
    \hspace{2em}
    \sum_{r=j}^n b_{ir} = u(o_i|O_{i-1} \cup L_{j-1})
    \label{eq:identities}
    \hspace{2em}
    \sum_{\substack{s=1\\s\neq i}}^n b_{sr}
    = u(l_r|L_{r-1}) - b_{ir}
\end{align}

Then Equation~(\ref{eq:constraint})
is equivalent to
\begin{align}\label{eq:localrule}
    \sum_{r=j}^n c(L_r) &u(l_r|L_{r-1})
    \leq [c(o_i) + c(L_{j-1})]
    {u(o_i|O_{i-1}\cup L_{j-1})}
    + \sum_{r=j}^n [c(o_i) + c(L_r)]
    [u(l_r|L_{r-1}) - b_{ir}].
\end{align}

Notice that Equation~(\ref{eq:localrule})
is not \textit{quite}
equivalent to the property of local optimality
with respect to insertions.
Instead, we know the following (very similar)
inequality holds
by the property that $L$ is locally
optimal:
\begin{align}
    \label{eq:actualrule}
    \sum_{r=j}^n c(L_r) &u(l_r|L_{r-1})
    \leq [c(o_i) + c(L_{j-1})]
    {u(o_i|L_{j-1})}
    + \sum_{r=j}^n [c(o_i) + c(L_r)]
    {u(l_r|L_{r-1} \cup \{o_i\})}.
\end{align}

We will now show that the right-hand side
of Equation~(\ref{eq:actualrule}) lower bounds
the right-hand side of Equation~(\ref{eq:localrule}).
Then Equation~(\ref{eq:localrule}) follows from
Equation~(\ref{eq:actualrule}).
We do this through the following four conditions
combined with Observation \ref{obs:ineq}:
\begin{align}
    \label{eq:monotone}
    &c(L_{j-1}) \leq c(L_{r-1})
    \\
    \label{eq:submod}
    &u(o_i|L_{j-1})
    \geq u(o_i|O_{i-1} \cup L_{j-1})
    \\
    \label{eq:supmod}
    &u(l_r|L_{r-1} \cup \{o_i\})
    \leq u(l_r|L_{r-1}) - b_{ir} \\
    \label{eq:algebra}
    &u(o_i|L_{j-1})
    + \sum_{r=j}^n {u(l_r|L_{r-1} \cup \{o_i\})}
    = {u(o_i|O_{i-1}\cup L_{j-1})}
    + \sum_{r=j}^n
    [u(l_r|L_{r-1}) - b_{ir}]
\end{align}
for all $i,j,r \in [n]$ with $j \leq r \leq n$.
Equation~(\ref{eq:monotone}) holds by the monotonicity
of $c$ and
Equation~(\ref{eq:submod}) holds by the submodularity
of $u$.
Equation~(\ref{eq:supmod}) holds
because $u$ is second-order supermodular.
(This is where we use second-order supermodularity.)

Finally, Equation~(\ref{eq:algebra})
holds by the following argument:
Notice that the left-hand side
is equal to $u([n])-u(L_{j-1})$ since
we sequentially sum the marginal
utility of adding the next element
to our current chain.
Similarly, the right-hand side simplifies
to $u([n]) - u(L_{j-1})$
after cancellation using 
Equation~(\ref{eq:identities}).

Using the above, we apply Observation \ref{obs:ineq}
with $C_0=c(o_i) + c(L_{j-1})$,
$C_{r-j+1} = c(o_i) + c(L_{r-1})$,
$X_0=u(o_i|L_{j-1})$,
$X_{r-j+1} = u(l_r|L_{r-1} \cup \{o_i\})$,
$Y_0=u(o_i|O_{i-1} \cup L_{j-1})$,
and $Y_{r-j+1}=u(l_r|L_{r-1})-b_{ir}$.
This yields
\begin{align*}
    & [c(o_i) + c(L_{j-1})]
    u(o_i|L_{j-1})
    + \sum_{r=j}^n
    [c(o_i) + c(L_{r-1})]
    [u(l_r|L_{r-1} \cup \{o_i\})]  \\ 
    & \leq [c(o_i) + c(L_{j-1})]
    u(o_i|O_{i-1} \cup L_{j-1})
    + \sum_{r=j}^n 
    [c(o_i) + c(L_{r-1})]
    [u(l_r|L_{r-1})-b_{ir}].
\end{align*}
Then the right-hand side of 
Equation~(\ref{eq:localrule})
must be greater than or equal to
the right-hand side of 
Equation~(\ref{eq:actualrule}).
Therefore Lemma \ref{lemma:localrule}
follows by Equation~(\ref{eq:actualrule}).
\end{proof}

\begin{proof}[Proof of Theorem \ref{thm:4approx}]
We will now prove that a locally
optimal permutation with respect to insertions
is (at worst) a 4-approximation of the optimal permutation.
To do this, we will show that the same linear program 
used by Munagala et al.~\cite{munagala2005pipelined} to 
prove their 4-approximation for coverage functions
also applies to 
all utility functions that are monotone,
submodular and second-order supermodular.
The linear program in Munagala et al.~\cite{munagala2005pipelined}
is:
\begin{align}
    & \textrm{maximize } \sum_{j=1}^n c(L_j)
    \sum_{i=1}^n b_{ij} \\
    &\textnormal{ subject to }
    \sum_{i=1}^n c(O_i) \sum_{j=1}^n b_{ij} \leq 1
    \hspace{1em} \textnormal{and} \nonumber  \\
    & \sum_{r=j}^n c(L_r) \sum_{s=1}^n b_{sr}
    \leq
    (c(o_i) + c(L_{j-1})) \sum_{r=j}^n b_{ir}
    + \sum_{r=j}^n (c(o_i) + c(L_r))
    \sum_{\substack{s=1\\s\neq i}}^n b_{sr} 
    \hspace{1em} \forall i,j \in [n]
    \label{eq:lp4approx}
\end{align}
where the variables $b_{ij}$ are non-negative
for all $i,j \in [n]$.

The first constraint scales the variables in
the optimal permutation so
that the objective is the ratio of the
local permutation to the optimal permutation.
By Lemma \ref{lemma:localrule}, the second constraint must hold for locally optimal $L$, with the given values for $b_{ij}$.
As shown in Munagala et al.~\cite{munagala2005pipelined},
there is a feasible solution to the dual of the LP
with objective value 4 so, by strong
duality, the locally optimal permutation
achieves a 4-approximation to the optimal permutation.
\end{proof}

\subsection{Proof of Theorem \ref{thm:4epsapprox}: $(4+\epsilon)$-approximation}

In this section, we show that local search
reaches a $(4+\epsilon)$-approximation in reasonable time.
The proof is based on the following fact, which bounds the progress made in each ``round'' of the local search (i.e., in each iteration of the while loop in Algorithm~\ref{alg:local}).

\begin{fact}
    \label{fact:improvement}
    Let $\pi$ be a permutation of $[n]$ that is an $M$-approximation to
    Min-Sum Submodular Cover.  
    Then applying one round of local search to this permutation will
    either establish that it is a local optimum or yield a new permutation that is
    an $(M-\frac{M-4}{2n})$-approximation.
\end{fact}

\begin{proof}
In Munagala et al.~\cite{munagala2005pipelined},
they prove that Fact \ref{fact:improvement}
holds for a Pipelined Filter Ordering instance, i.e., when $u$ is a coverage
function.
Consider an $M$-approximate solution.
Let $A$ be a variable representing the
reduction in approximation factor after
making the local search step.
As in Equation~(\ref{eq:lp4approx}),
the variables $b_{ij}$ are non-negative
real numbers and correspond to a concept
of shared utility formalized in Lemma \ref{lemma:localrule}.
Then the following linear program is what
Munagala et al.~\cite{munagala2005pipelined} 
use to lower-bound the improvement in approximation ratio:
\begin{align*}
    &\textrm{minimize } A \\
    &\textrm{ subject to } 
    \sum_{i=1}^n c(O_i) \sum_{j=1}^n b_{ij} \leq 1
    \textrm{ and } \\
    &\sum_{r=j}^n c(L_r) \sum_{s=1}^n b_{sr}
    \leq A+
    (c(o_i) + c(L_{j-1})) \sum_{r=j}^n b_{ir}
    + \sum_{r=j}^n (c(o_i) + c(L_r))
    \sum_{\substack{s=1\\s\neq i}}^n b_{sr} 
    \textrm{ and }\\
    &\sum_{j=1}^n c(L_j) \sum_{i=1}^n b_{ij} \geq M
    \hspace{1em} \forall i,j \in [n].
\end{align*}
For the more general case where $u$ is an
abstract submodular and second-order supermodular
utility function, the first and third
inequalities trivially hold for our generalized
definition of $b_{ij}$.
The second inequality follows from
Lemma \ref{lemma:localrule}.
By taking the dual, Munagala et al.~\cite{munagala2005pipelined} show that the objective
of the primal and therefore the reduction in
approximation ratio is at least $\frac{M-4}{2n}$
which gives Fact \ref{fact:improvement}.
When $u$ is an arbitrary submodular and second-order
supermodular function, 
it follows from Lemma \ref{lemma:localrule}
that the same linear program still lower-bounds the improvement.
\end{proof}

To achieve a good bound on the time required for the local search for Min-Sum Submodular Cover, we want to begin the local search from a permutation that is not too far from optimal.  The 
following theorem gives such a permutation.

\begin{theorem}[Non-Decreasing Cost]
\label{thm:cost}
A permutation that orders the elements $i\in [n]$ in non-decreasing order of $c(\{i\})$ is an $n$-approximate solution to Min-Sum
Submodular Cover.
\end{theorem}

\begin{proof}
Suppose, without loss of generality, that 
$c(\{1\}) \le \cdots \le c(\{n\})$.

Suppose an optimal solution is given by the sets $\emptyset=S_0,S_1 \ldots, S_n$, and for each $j \in [n]$, let $S'_j=[\max S_j]$, where $\max S_j$ is the maximum integer in $S_j$. For instance, for $n=4$ and
$S_1=\{2\}$,
$S_{2}=\{12\}$,
$S_3=\{124\}$,
$S_4=\{1234\}$, we have 
$S'_1=\{12\}$,
$S'_2=\{12\}$,
$S'_3=\{1234\}$,
$S'_4=\{1234\}$.
Setting $[0]$ and $S'_0$ equal to $\emptyset$,
observe that
\begin{align*}
\sum_{j=1}^n c([j])(u([j])-u([j-1]))
 \le \sum_{j=1}^n c(S'_j)(u([j])-u([j-1]))
= \sum_{j=1}^n c(S'_j)(u(S'_j)-u(S'_{j-1})),
\end{align*}
where the inequality follows by the monotonicity of
$c$ and the equality follows (not term-wise but in total)
by considering when
utility is accrued by each side.

Now, for each $j \ge 0$, we have $c(S_j) \ge c(\{\max S_j\})$, by the monotonicity of $c$. Also, since $c(\{\max S_j\}) \ge c(\{i\})$ for every $i \in S_j$ by the indexing assumption,
\[
n \cdot c(S_j) \ge n \cdot c(\{\max S_j\}) = n \cdot c(\{\max S'_j\}) \ge  c(S'_j)
\]
where the equality follows from the definition of $S'_j$
and the second inequality follows since there are at most $n$
elements in $S'_j$.
Then the objective value of the optimal permutation is at least ${\frac{1}{n} \sum_{j=1}^n c(S'_j)(u(S_j)-u(S_{j-1}))}$, or equivalently by charging utility to each increase in cost,
\[
\frac{1}{n} \sum_{j=1}^n (u([n])-u(S_{j-1}))(c(S'_j)-c(S'_{j-1})).
\]
By the monotonicity of $u$, this sum is at least
\[
\frac{1}{n} \sum_{j=1}^n (u([n])-u(S'_{j-1}))(c(S'_j)-c(S'_{j-1}))
\Leftrightarrow
\frac{1}{n} \sum_{j=1}^n c(S'_j)(u(S'_j)-u(S'_{j-1})).
\]
This is at least $1/n$ times the objective value of the increasing cost permutation, by our earlier observation.
\end{proof}

We note that Theorem \ref{thm:cost} also holds
for Min-Sum Permutation
Problems minimizing over all permutations where
$u$ only satisfies monotonicity and $c$ only satisfies
monotonicity and subadditivity (both are still
normalized).

We can now prove that the output of Algorithm \ref{alg:local}
is a $(4+\epsilon)$-approximation to Min-Sum Submodular Cover. 
The time bound assumes constant time oracle queries.


\begin{proof}[Proof of Theorem \ref{thm:4epsapprox}]
The improvement in the quality of the solution in each round of local search, guaranteed by Fact~\ref{fact:improvement}, implies 
that a $(4+\epsilon)$-approximation is achieved within 
$O(n \log (\frac{d}{\epsilon}))$ rounds of Algorithm~\ref{alg:local}, when it is initialized with a $d$-approximate permutation.
This implication was stated without proof by Munagala et al.~\cite{munagala2005pipelined}, in proving the same bounds for Pipelined Set Cover.  For completeness, we present a proof of the implication in Appendix~\ref{app:mainproof}.  
The $O(n^3 \log (\frac{d}{\epsilon}))$ time bound is achieved by spending $O(n^2)$ per round.
To accomplish this, we do not recompute all $n$ terms of the objective function for each of the $\theta(n^2)$ neighbors of the current solution.
Instead, by considering `adjacent' neighbors
sequentially, we can compute the objective function value for the
next neighbor from the value obtained for the previous neighbor in constant time,
by recomputing only two terms of the objective function.
The time bound for the non-decreasing cost permutation follows from Theorem~\ref{thm:cost}.
\end{proof}



\section{Experiments}\label{sec:experiments}

Assuming constant-time oracle access to the utility functions,
the greedy algorithm runs in $O(n^2)$ total time, while our
local search algorithm spends time $O(n^2)$ in each round.
In our experiments, with no oracle, we had to compute the utility function.

The greedy algorithm is certainly faster than the local search, but
it only explores one type of solution, where
cost effective elements appear earlier in the permutation. 
Local search initialized with random permutations
can sample from the entire solution space.
Our experiments compare the greedy solution
to local search solutions from four random initial permutations,
and from a non-decreasing cost permutation.  We run each local search
for $n$ steps
(rather than running it to convergence, or until it is guaranteed to find a $(4+\epsilon)$-approximate solution).
We see empirically that the best of the 5 local search solutions
tends to be better than the worst, and also better than 
the greedy solution.
In applications where computing is cheap, $n$ is not large,
and the quality of the solution is crucial,
using local search may be preferable to using greedy.

Our experiments compare greedy and local search 
on 100 random instances of two problems: Pipelined Set Cover  and Min-Sum Facility Location.
For our random instances, we set $n=30$ and each cost 
$c_i$ to be a uniform random value between 0 and 1.  
Figure \ref{fig:experiments} shows the results of our experiments, with objective values given relative to the best 
of the 6 solutions (1 greedy and 5 local search).
Local search finds the best solution in almost $100\%$
of the 100 instances whereas greedy finds the best
in roughly $50\%$.

\begin{figure}
    \centering
    \includegraphics[width=.45\textwidth]{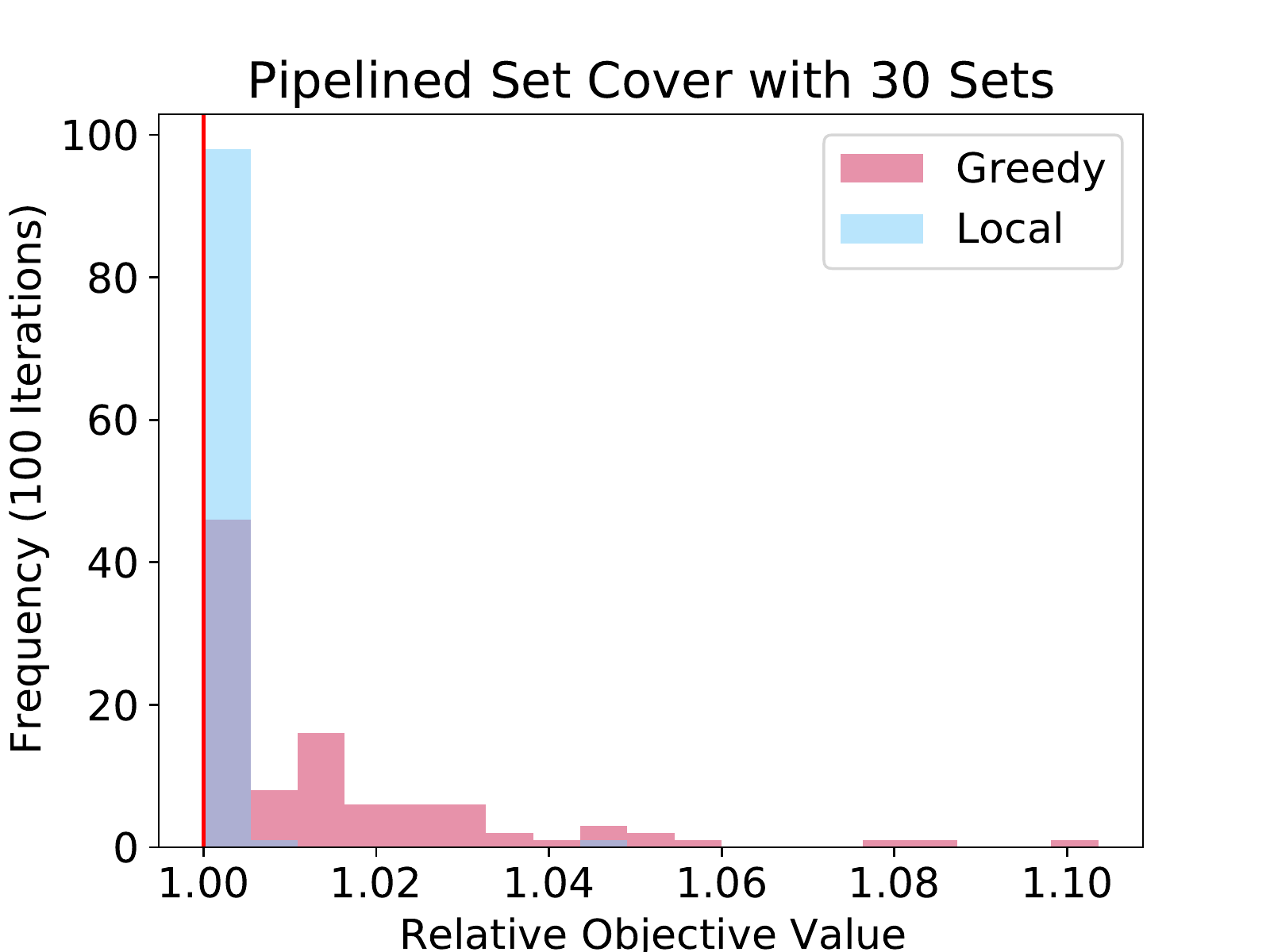}
    \includegraphics[width=.45\textwidth]{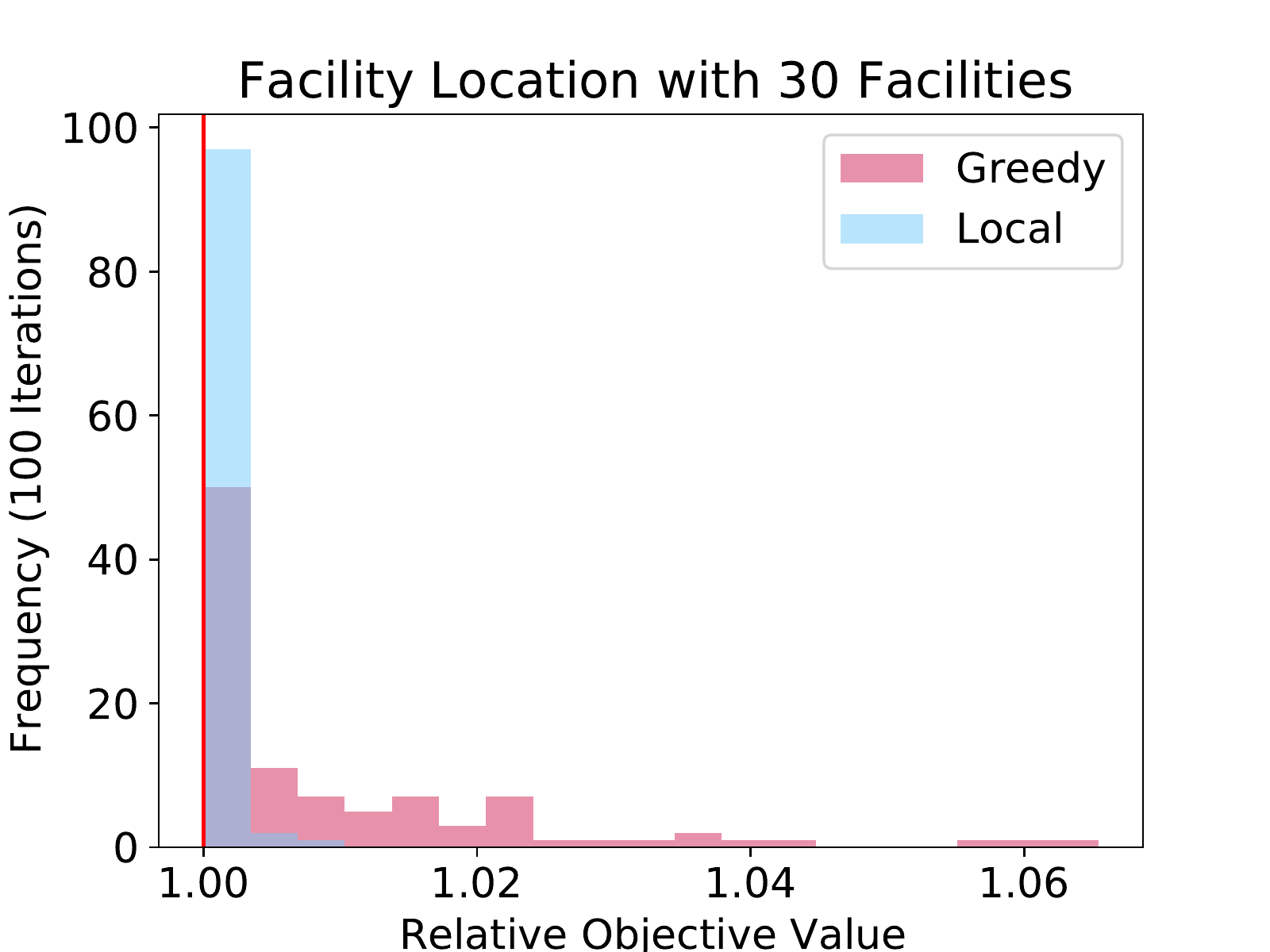}
    \caption{Histograms of greedy and local search performance
    with $n=30$ in Pipelined Set Cover and 
    Min-Sum Facility Location.
    The relative objective values are normalized with respect to the best
    of the 6 generated solutions (1 greedy and 5 local search). Frequency is reported from the 100 random instances generated for each dataset.}  
    \label{fig:experiments}
\end{figure}

\paragraph*{Pipelined Set Cover}
We perform experiments on synthetic, randomly generated instances of the (unweighted) Pipelined Set Cover problem
with correlated subsets, following an approach of Babu et al.~\cite{babu2004adaptive}.
Recall that an instance of (unweighted) Pipelined Set Cover consists of a finite ground set $[m]$ and
a collection of subsets 
$D_i \subseteq V$ for $i \in [n]$.
The utility of a set $S \subseteq [n]$ is $u(S) = |\bigcup_{i \in S} D_i|$.
The $n$ subsets $D_i$ in our random instance are divided
into $\lceil n/\Gamma \rceil$ groups, where $\Gamma$ is a ``correlation factor.''
The instance has the following properties,
for each element $j \in [m]$.
For all $i \in [n]$, $\Pr[j \in D_i]$ is a fixed value $p$.
For two subsets $D_i$ and $D_{i'}$
in different groups, 
membership of $j$ in $D_i$
is independent of its membership in $D_{i'}$.
For two subsets $D_i$ and $D_{i'}$
in the same group, the probability that $j$ has the same
membership status in $D_i$ and $D_{i'}$ (i.e., is either in both subsets,
or in neither), is a fixed value $\rho$.
In our experiments, $m=2n$, $\Gamma=4$, $p=0.3$, and $\rho =0.7858$.
In Appendix~\ref{app:babuexp}, we describe in detail how we generated the instance.

\paragraph*{Min-Sum Facility Location}
We use the locations of $n$ Citi Bike stations \cite{citibike}
in New York City as the facilities for our 
facility location data set.
For calculating the utility to customers, we uniformly
generate customer locations within the range of latitude
and longitude of the stations.
The value $M_{a,b}$ for customer $b$ and station $a$ is
the inverse of the Euclidean distance between them.

\bibliography{references}

\appendix
\section{Appendix}
\subsection{Prior Use of the Term `Min-Sum Submodular Cover'}\label{app:prior}
The first proof that a greedy algorithm for Min-Sum Submodular Cover yields a 4-approximation was given by Streeter and Golovin.  The proof appears both in a Technical Report~\cite{streeter2007online} and an associated conference paper~\cite{streetergolovin}. They actually gave their proof for a more general problem than the one we considered in this paper, 
in which $u:2^{[n] \times \R} \rightarrow \R$, and the output is a sequence of pairs of the form $(i,\tau) \in [n] \times \R$.
In their Technical Report, Streeter and Golovin used the name Min-Sum Submodular Cover to refer to the more general problem, but they did not use this name (nor any other) to refer to the problem in their conference paper. We opted to use the name Min-Sum Submodular Cover to refer to the problem we defined in this paper, as we believe this usage of the name is natural given the connection to Min-Sum Submodular Cover.  

We note that the definition Streeter and Golovin gave for the more general problem is problematic as written. 
The greedy algorithm may not be well-defined for functions $u$ that are non-zero for subsets of $[n] \times \R$
that include pairs $(i,\tau)$, where $\tau$ is infinitesimally small. However, the results in the paper are not dependent on allowing such $u$, and the problem with the definition can be fixed by restricting the domain of $u$.

\subsection{Bounding the number of rounds in the proof of Theorem~\ref{thm:4epsapprox}} \label{app:mainproof}

We show using Fact \ref{fact:improvement} that Algorithm~\ref{alg:local}
yields a $(4+\epsilon)$-approximation
in at most $O(n \log(\frac{d}{\epsilon}))$
rounds from a $d$-approximate permutation.

We introduce a recurrence relation 
$T(\ell) = a T(\ell-1) - b$
where $a=2n/(2n-1)$ and $b=4/(2n-1)$.
We can derive this recurrence by setting
$T(\ell)=M$ and $T(\ell-1)=M-(M-4)/2n$.
Intuitively, $\ell$ is the number of iterations until
we reach $(4+\epsilon)$
and so we set $T(0)=4+\epsilon$.
By repeatedly expanding $T(\ell)$, we get
\begin{align*}
    T(\ell)=-b\sum_{i=0}^{\ell-1} a^i + a^\ell (4+\epsilon) 
    =-b\frac{a^\ell-1}{a-1} + a^\ell (4+\epsilon)
    =\epsilon \left(\frac{2n}{2n-1}\right)^\ell + 4
\end{align*}
where the last equality follows by plugging
in the values of $a$ and $b$.
We claim that a $d$-approximate permutation 
is at most $2n \log(\frac{d}{\epsilon})$ rounds from a
$(4+\epsilon)$-approximation.
We can verify this by evaluating
\begin{align*}
    T\left(2n \log\left(\frac{d}
    {\epsilon}\right)\right) = \epsilon 
    \left(\frac{2n}{2n-1}\right)
    ^{2n \log(\frac{d}{\epsilon})} + 4 
    = d^{2n \log\left(\frac{2n}{2n-1}\right)}+4.
\end{align*}
Notice that
\begin{align*}
    x \log\left(\frac{x}{x-1}\right) \geq 1 \iff
    e^x \left(\frac{x}{x-1}\right) \geq e
\end{align*}
which is certainly true for $x>2$.
It follows that $T(2n \log(\frac{d}{\epsilon})) \geq d$
so local search converges in at most
$2n \log(\frac{d}{\epsilon})$ rounds.

\subsection{Generation of the Pipelined Set Cover Data Set} \label{app:babuexp}

We generate the $n$ subsets $D_1, \ldots, D_n$ of $[m]$
randomly as follows.
Initially, for each group $G$ of subsets, we generate 
an advice bit $a_{G,j}$ for each element $j \in [m]$, which
is True with probability $p$, and False with probability $1-p$.
Then for each element $j$, and for each $D_i$ in group $G$, we 
do the following: with probability $p'$ we use advice bit $a_{G,j}$ to determine whether or not to include $j$ in $D_i$ (if $a_{G,j}$ is True, we include $j$ in $D_i$, else we do not).  With probability $(1-p')$, we ignore the advice bit, and instead,
we include $j$ in $D_i$ with probability $p$, and exclude it with probability $(1-p)$.  The probability that $j$ is in a given set $D_j$ is clearly $p$.

For $i \neq i'$, if subsets $D_i$ and $D_{i'}$ are in different groups, then membership of an element $j$ in $D_i$ is clearly independent of its membership of $D_{i'}$.
If $D_i$ and $D_{i'}$ are subsets in the same group $G$, then
the probability that $j$ has the same membership status in both subsets can be calculated by noting that this can happen in two ways:
either the advice bit $a_{G,j}$ was used to determine membership in both subsets, or $a_{G,j}$ was ignored for one or both of the two subsets and membership ended up being the same in both subsets.
The first event happens with probability $p'\cdot p'$.  
The second happens with probability $(1-p'\cdot p')(p\cdot p + (1-p)\cdot (1-p))$.
Because $p=0.3$ and $p'=0.7$, the probability that the membership status of $j$ is the same for both subsets is $.7858$.

We note that it is possible that this process results in subsets $D_i$ and $D_{i'}$ where $i \neq i'$ and $D_i = D_{i'}$.  We do not eliminate such duplicates.

\end{document}